\documentclass[conference]{IEEEtran}
\usepackage[dvips]{graphicx}
\usepackage{amsfonts}
\usepackage{amssymb}
\usepackage{amsmath}
\usepackage{color}
\usepackage{ booktabs}
\usepackage{array}  

\newtheorem{Theorem}{Theorem}
\newtheorem{Remark}{Remark}

\newtheorem{Corollary}{Corollary}

\newcommand{\set}{\mathcal}

\ifCLASSINFOpdf

\else

\fi

\begin{document}

\title{Achievable Rate Regions for Cooperative  Relay Broadcast Channels with Feedback}

\author{
  \IEEEauthorblockN{Youlong Wu\\}
  \IEEEauthorblockA{Lehrstuhl f\"ur Nachrichtentechnik\\
    Technische Universit\"at M\"{u}nchen, Germany\\
     youlong.wu@tum.de }

}


\maketitle

 \begin{abstract}
 Achievable rate regions for  cooperative relay broadcast channels with rate-limited feedback  are proposed. Specifically, we consider two-receiver  memoryless broadcast channels where each receiver  sends feedback signals  to the transmitter through a noiseless and rate-limited feedback link,  meanwhile, acts as a relay to transmit cooperative information to the other receiver. It's shown that the proposed rate regions improve on the known regions that  consider either relaying cooperation or feedback communication, but not both.
 \end{abstract}


\section{Introduction}
Relay broadcast channels (RBCs) describe  communication networks where the transmitter sends information to a set of receivers with the help of relaying communication. In \cite{Dabora'06}, \cite{Kramer'05}, the dedicated-relay broadcast channel (BC) model  was studied, where a relay node was introduced to the original BC to assist the cooperation between  two receivers.  Another RBC model, called  cooperative RBC model was studied in  \cite{Liang'07Veeravalli}, \cite{Liang'07Kramer}, where each receiver acts as a relay and sends cooperative information to the other receiver. It was shown that even partially cooperation (only one receiver relays cooperative information)  still  improves on the capacity region of original BC.

In a different line of work,  many studies have been done on  memoryless BCs with feedback, where the receivers  send feedback signals to the transmitter through feedback links. In \cite{gamal'78}, it shows that feedback cannot increase the capacity region for all physically degraded BCs.  The first simple example BC where feedback  increases capacity was presented by Dueck \cite{dueck}. Based on Dueck's idea, Shayevitz and Wigger \cite{wigger}  proposed an achievable region for BCs with generalized feedback.  
 Other achievable regions for   BCs with perfect or noisy feedback, have been proposed by Kramer \cite{kramer} and  Venkataramanan and Pradhan \cite{venkataramananpradhan11}.  Most recently, Wu and Wigger \cite{Wu'ISIT14,YoulongArxiv} showed that any positive feedback rate can increase the capacity region for a large class of BCs, called strictly essentially less noisy BCs, unless it is physically degraded.  
  
Cooperative RBCs with prefect feedback was investigated in \cite{Liang'07Veeravalli}, where the capacity region was established for the case of perfect feedback from the receiver to the relay.
 In this paper, we consider the cooperative  RBCs with rate-limited feedback from the receivers to the transmitter, i.e., each receiver  sends feedback signals  to the transmitter through a noiseless and rate-limited feedback link, and meanwhile, acts as a relay to transmit cooperative information to the other receiver. 
 	
 In the first work, we first  study the \emph{partially} cooperative RBC with one-sided feedback (only one receiver sends feedback signals and relays cooperative information to the other receiver). We proposed a new coding scheme (Scheme 1) based on block-Markov coding, Marton's coding \cite{Marton'79}, partial decode-forward  \cite{Cover'79} and compress-forward strategies \cite{Cover'79}. Specifically, in each block, the transmitter uses Marton's coding to send the source messages and \emph{forward} the feedback message. The receiver who acts as relay  performs combined partial decode-forward and compress-forward, and sends the compression message as feedback information. The other Receiver  uses backward decoding to jointly decode its private message and the compression message.  It is shown that when feedback rate is sufficiently large, our coding scheme strictly improves on Liang and Kramer's region \cite{Liang'07Kramer}, which is tight for the  semideterministic partially cooperative RBCs and orthogonal partially cooperative RBCs. 
 
In the second work, we  study the \emph{fully} cooperative RBCs with two-sided feedback (both  receiver send feedback signals and relay cooperative information). Two block-Markov coding schemes (Scheme 2A and 2B) are proposed based on Scheme 1. Specifically,  in each block, the transmitter uses Marton's coding to send the source messages and forward the feedback messages sent by both receivers. In Scheme 2A, both receivers apply compress-forward and backward decoding.  Scheme 2B is similar to Scheme 1A except that one of the two receiver uses hybrid relaying strategy and sliding-window decoding.  The resulting rate regions strictly improve on Wu and Wigger's region  \cite[Theorem 1]{Wu'ISIT14}, which shows that   feedback  strictly  increases capacity region for a large class of BCs.
  
  Note that in our coding schemes the transmitter can reconstruct the  receivers' inputs due to a delicate design,  which allows to superimpose the Marton's codes on the receivers inputs, and thus attains cooperation between the transmitter and the receivers.
  
This paper is organized as follows.  Section \ref{Sec:Model} describes   cooperative RBC with feedback and  our main results are presented in Section \ref{Sec:results}. Section \ref{Sec:compare} compares various achievable rate regions and shows that our   regions strictly improve the known rate regions that consider either relay cooperation or feedback communication, but not both.  Sections \ref{sec:FullCooperation} and \ref{Sec:ParRBC} contain the proofs of our results in Section \ref{Sec:results}.  Finally, Section  \ref{Sec:conclusion} concludes this paper.

Notations: 
 We use capital letters to denote  random variables and small letters for their realizations, e.g. $X$ and $x$. For  $j\in\mathbb{Z}^+$, we use the short hand notations $X^j$ and $x^j$ for the tuples $X^j:=(X_1,\ldots, X_j)$ and $x^j:=(x_1,\ldots, x_j)$.  Given a positive integer $n$, let $\mathbf{1}_{[n]}$ denote the all-one tuple of length $n$, e.g., $\mathbf{1}_{[3]}=(1,1,1)$. The abbreviation i.i.d. stands for \emph{independent and identically distributed}.

Given a distribution $P_A$ over some alphabet $\set{A}$, a positive real number $\varepsilon>0$, and a positive integer $n$, let $\set{T}_{\varepsilon}^n(P_A)$ denote the typical set in \cite{book:gamal}. 
 
\section{System model}\label{Sec:Model}

Consider 3-node  cooperative RBC  with  feedback, as shown in Fig. \ref{fig:RBC_full}. 
  This setup is characterized by seven  finite alphabets $\set{X}, \set{X}_k, \set{Y}_k, \set{F}_k$, for $k\in\{1,2\}$,  a channel law $P_{Y_1Y_2|XX_1X_2}$ and nonnegative feedback rates $R_{\text{fb},1},R_{\text{fb},2}$. Specifically, at  discrete-time $i\in[1:n]$, the transmitter  sends  input $x_{i}\in \set{X}$.  Receiver $k$ observes  output  $y_{k,i}\in\set{Y}_k$ and relays cooperative information $x_{k,i}\in\set{X}_k$ to the other receiver. When both receiver relay information, it is called \emph{fully} cooperative RBC. When only one receiver relays information, it is called \emph{partially} cooperative RBC. After observing $y_{k,i}$, Receiver $k$ also sends a feedback signal $f_{k,i}\in \set{F}_{k,i}$ to the transmitter, where $\set{F}_{k,i}$ denotes the finite alphabet of $f_{k,i}$. The feedback link between  the transmitter and Receiver $k$ is  noiseless and \emph{rate-limited} to $R_{\textnormal{fb},k}$ bits per channel use. In other words, if the transmission takes place over a total blocklength $n$, then
\begin{IEEEeqnarray}{rCl}\label{consFB0}
|\set{F}_{k,1}|&\times&\cdots\times|\set{F}_{k,n}|\leq 2^{nR_{\textnormal{fb},k}}, \quad k\in\{1,2\}.
\end{IEEEeqnarray}
{In the communication, the transmitter  wishes to send message $M_0\in[1:2^{nR_0}]$ to both receivers, and message $M_k\in[1:2^{nR_k}]$ to Receiver  $k$. }

\begin{figure}
\centering
\includegraphics[width=0.45\textwidth]{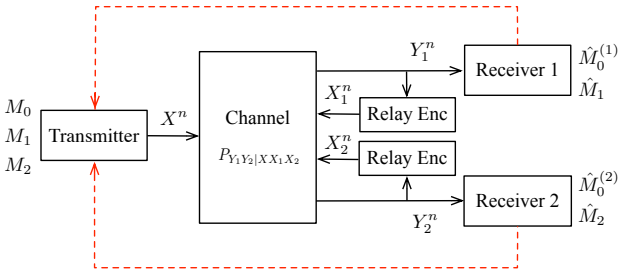}
\caption{Cooperative relay broadcast channel with feedback} \label{fig:RBC_full}
\end{figure}


 A $(2^{nR_0},2^{nR_1},2^{nR_2},n)$  code for this channel consists of 
\begin{itemize}
\item  message sets $\set{M}_0:=[1:2^{nR_0}]$ and $\set{M}_k:=[1:2^{nR_k}]$;
\item a source encoder that maps $(M_0,M_1,M_2)$ to a sequence $X_{i}\big(M_0,M_1,M_2,F^{i-1}_1,F^{i-1}_{2}\big)$; 
\item two receiver  encoders where Receiver $k$ maps $Y^{i-1}_k$ to a sequence $X_{k,i}(Y^{i-1}_k)$; 
\item two decoders where  Receiver $k$ estimates $(\hat{M}^{(k)}_0,\hat{M}_k)$ based on $Y^n_k$,
\end{itemize}
for each time $i\in[1:n]$ and $k\in\{1,2\}$.
Suppose $M_0,M_1$ and $M_2$ are uniformly distributed and independent with each other. A rate tuple $(R_0,R_1,R_2)$ with average feedback rates $R_{\textnormal{fb},k}$, for $k\in\{1,2\}$,  is called achievable if for every blocklength $n$, there exists  a $(2^{nR_0},2^{nR_1},2^{nR_2},n)$ code such that the average probability of error
\begin{IEEEeqnarray*}{rCl}
P^{(n)}_e=\text{Pr}[(\hat{M}_0^{(1)},\hat{M}_0^{(2)},\hat{M}_1,\hat{M}_2 )\neq (M_0,M_0,M_1,M_2)]
\end{IEEEeqnarray*}
tends to 0 as $n\to\infty$.   The capacity region is all nonnegative  rate tuples $(R_0,R_1,R_2)$ such that  $\lim_{n\to \infty }P_e^{(n)}=0$. 


\section{Main Results}\label{Sec:results}
In this section, we present our main results as the following   theorems. The  proofs are given in Section \ref{Sec:ParRBC} and Section \ref{sec:FullCooperation}.

\begin{Theorem}\label{Them:Scheme1}
 For the  \emph{partially} cooperative BRC with receiver-transmitter feedback, the  capacity region includes the set $\set{R}_1$ of all nonnegative rate tuples $(R_0,R_1,R_2)$ that satisfy 
 \begin{subequations} \label{eq:regionSch1}
\begin{IEEEeqnarray}{rCl}
R_0+R_1&\leq &I(U_0,U_1;Y_1|X_1) \\
R_0+R_2 &\leq & I(U_0,U_2,X_1;Y_2)\nonumber\\
&&-I(\hat{Y}_1;Y_1|U_0,U_2,X_1,Y_2)\\
R_0+R_1+R_2&\leq &  I(U_1;Y_1|U_0,X_1)+I(U_0,U_2,X_1;Y_2)\nonumber\\
&&-I(\hat{Y}_1;Y_1|U_0,U_2,X_1,Y_2)\nonumber\\&&-I(U_1;U_2|U_0,X_1)\\
R_0+R_1+R_2&\leq &I(U_0,U_1;Y_1|X_1) + I(U_2;\hat{Y}_1,Y_2|U_0,X_1)\nonumber\\
&&-I(U_1;U_2|U_0,X_1)\\
2R_0+R_1+R_2&\leq &I(U_0,U_1;Y_1|X_1)+ I(U_0,U_2,X_1;Y_2)\nonumber\\
&&-I(\hat{Y}_1;Y_1|U_0,U_2,X_1,Y_2)\nonumber\\
&&-I(U_1;U_2|U_0,X_1)
\end{IEEEeqnarray}
for some pmf ${ P_{U_0U_1U_2X_1}}P_{\hat{Y}_1|U_0X_1Y_1}$ and function $X=f(U_0,U_1,U_2)$ such that  
\begin{IEEEeqnarray}{rCl}\label{FbrateTh1}
&&I(\hat{Y}_1;Y_1|U_0,X_1)\leq R_{\text{fb},1}.
\end{IEEEeqnarray}
 \end{subequations}
\end{Theorem}

\begin{proof}
See Section \ref{Sec:ParRBC}.
\end{proof}

\begin{Remark}\label{Remark:Wyner-Ziv}
The rate constraint \eqref{FbrateTh1} can be relaxed as
\begin{IEEEeqnarray}{rCl}
&&I(\hat{Y}_1;Y_1|U_0,X_1,Y_2)\leq R_{\text{fb},1}
\end{IEEEeqnarray}
by using a trick in \cite[Section V]{YoulongArxiv}, where  the receivers use the feedback links to send Wyner-Ziv compression messages about their previously observed outputs to the transmitter.  
\end{Remark}

\begin{Remark}\label{Remark:Liang}
If $\hat{Y}_1=\emptyset$, i.e.,  no feedback signal is sent by Receiver 1, then rate region $\set{R}_1$  reduces to $\set{R}_{\text{Liang}}$, which is the set of all nonnegative rate tuples $(R_0,R_1,R_2)$ satisfying
 \begin{subequations} \label{rate:Liang}
\begin{IEEEeqnarray}{rCl}
R_0+R_1&\leq &I(U_0,U_1;Y_1|X_1) \\
R_0+R_2 &\leq & I(U_0,U_2,X_1;Y_2)\\
R_0+R_1+R_2&\leq &  I(U_1;Y_1|U_0,X_1)+I(U_0,U_2,X_1;Y_2)\nonumber\\&&-I(U_1;U_2|U_0,X_1)\\
R_0+R_1+R_2&\leq &I(U_0,U_1;Y_1|X_1) + I(U_2;Y_2|U_0,X_1)\nonumber\\
&&-I(U_1;U_2|U_0,X_1)\\
2R_0+R_1+R_2&\leq &I(U_0,U_1;Y_1|X_1)+ I(U_0,U_2,X_1;Y_2)\nonumber\\
&&-I(U_1;U_2|U_0,X_1)
\end{IEEEeqnarray}
\end{subequations}
for some pmf ${ P_{U_0U_1U_2X_1}}$ and function $X=f(U_0,U_1,U_2)$. This rate region was proposed by Liang and Kramer \cite[Theorem 2]{Liang'07Kramer}, and was shown to be the capacity region for  semideterministic partially cooperative RBCs and orthogonal partially cooperative RBCs. 
\end{Remark}

\begin{Theorem}\label{Them:Scheme2A}
 For the \emph{fully} cooperative BRC  with two-sided and rate-limited feedback, the  capacity region includes the set $\set{R}_2$ of all nonnegative rate tuples $(R_0,R_1,R_2)$ that satisfy 
 \begin{subequations} \label{eq:Scheme2A}
\begin{IEEEeqnarray}{rCl}
R_0+R_1&\leq &I(U_0,U_1;\hat{Y}_2,Y_1|X_1,X_2)+\Delta_1 \\
R_0+R_2 &\leq & I(U_0,U_2;\hat{Y}_1,Y_2|X_1,X_2)+\Delta_2 \\
R_0+R_1+R_2&\leq & I(U_0,U_1;\hat{Y}_2,Y_1|X_1,X_2)+\Delta_1\nonumber\\
&&+I(U_2;Y_2,\hat{Y}_1|U_0,X_1,X_2)\nonumber\\
&&-I(U_1;U_2|U_0,X_1,X_2)\\
R_0+R_1+R_2&\leq & I(U_0,U_2;\hat{Y}_1,Y_2|X_1,X_2)+\Delta_2\nonumber\\
&&+I(U_1;Y_1,\hat{Y}_2|U_0,X_1,X_2)\nonumber\\
&&-I(U_1;U_2|U_0,X_1,X_2)\\
2R_0+R_1+R_2&\leq & I(U_0,U_1;\hat{Y}_2,Y_1|X_1,X_2)+\Delta_1\nonumber\\
&&+I(U_0,U_2;\hat{Y}_1,Y_2|X_1,X_2)+\Delta_2\nonumber\\
&&-I(U_1;U_2|U_0,X_1,X_2)
\end{IEEEeqnarray}
for some pmf $P_{X_1}P_{X_2}{ P_{U_0U_1U_2|X_1X_2}}P_{\hat{Y}_1|X_1Y_1}P_{\hat{Y}_2|X_2Y_2}$ and function $X=f(U_0,U_1,U_2)$ such that 
\begin{IEEEeqnarray}{rCl}
&&I(\hat{Y}_1;Y_1|X_1)\leq R_{\text{fb},1}\label{Fbrate1Th2} \quad \textnormal{and} \quad I(\hat{Y}_2;Y_2|X_2)\leq R_{\text{fb},2}\label{Fbrate2Th2}
\end{IEEEeqnarray}
 where
\begin{IEEEeqnarray}{rCl} 
&&\Delta_1=\min\{0,I(X_{2};Y_1|X_1)\!-\!I(\hat{Y}_{2};Y_{2}|X_1,X_2,Y_1)\}\nonumber\\
&&\Delta_2=\min\{0,I(X_{1};Y_2|X_2)\!-\!I(\hat{Y}_{1};Y_{1}|X_1,X_2,Y_2)\}.\nonumber
\end{IEEEeqnarray}
 \end{subequations}
\end{Theorem}

\begin{proof}
See Section \ref{sec:scheme2A}.
\end{proof}

\begin{Remark}\label{Remark:WU}
If $R_0=0$ and ${X}_1=X_2=\emptyset$, i.e., both receivers send feedback signals without relaying cooperative information, by relaxing rate constraint \eqref{FbrateTh1} as in Remark \ref{Remark:Wyner-Ziv}, the rate region $\set{R}_1$ reduces to $\set{R}_\text{Wu}$, which is the set of all nonnegative rate tuples $(R_0,R_1,R_2)$ satisfying 
\begin{subequations}\label{eq:region_relay}
\begin{IEEEeqnarray}{rCl}
R_1 &\leq& I(U_0,U_1;Y_1,\hat{Y}_2)-I(\hat{Y}_2;Y_2|Y_1)\label{In2R1}\\
R_2 &\leq& I(U_0,U_2;Y_2,\hat{Y}_1)-I(\hat{Y}_1;Y_1|Y_2) \label{In2R2}\\
R_1\!+\!R_2 &\leq& I(U_0,U_1;Y_1,\hat{Y}_2) -I(\tilde{Y}_2;Y_2| Y_1)\nonumber \\
& & + I(U_2;Y_2,\hat{Y}_1|U_0) - I(U_1;U_2|U_0)\label{In2Sum1}\\
R_1\!+\!R_2 &\leq& I(U_0,U_2;Y_2,\hat{Y}_1)-I(\hat{Y}_1;Y_1|Y_2)\nonumber \\ 
& &+ I(U_1;Y_1,\hat{Y}_2|U_0) - I(U_1;U_2|U_0)\label{In2Sum2}\IEEEeqnarraynumspace\\
R_1\!+\!R_2 &\leq&  I(U_0,U_1;Y_1,\hat{Y}_2) -I(\hat{Y}_2;Y_2|Y_1)- I(U_1;U_2|U_0)  \nonumber \\ 
& & +I(U_0,U_2;Y_2,\tilde{Y}_1)-I(\hat{Y}_1;Y_1|Y_2) \label{In2Sum3}
\end{IEEEeqnarray}
for some pmf $P_{U_0U_1U_2}P_{\tilde{Y}_1|Y_1}P_{\tilde{Y}_2|Y_2}$ and function $X=f(U_0,U_1,U_2)$ such that
\begin{IEEEeqnarray}{rCl}
I(\tilde{Y}_1;Y_1|Y_2)&\leq&    R_{\textnormal{Fb},1} \quad \textnormal{and} \quad 
I(\tilde{Y}_2;Y_2|Y_1)\leq    R_{\textnormal{Fb},2}.\IEEEeqnarraynumspace
  \end{IEEEeqnarray}
\end{subequations}
  This rate region coincides with Wu and Wigger's region in \cite[Corrollary 1] {Wu'ISIT14}, which shows feedback can strictly increase the entire capacity region for a large class of BCs, called strictly essentially less noisy BCs, unless it is physically degraded. 
\end{Remark}

In the scheme for Theorem \ref{Them:Scheme2A}, both receivers apply compress-forward. If one of the two receivers uses a hybrid relaying strategy that combines partially decode-forward and compress-forward, we obtain a new achievable region below. 

\begin{Theorem}\label{Them:Scheme2B}
 For the \emph{fully} cooperative  BRC with {two-sided} and rate-limited feedback, 
the  capacity region includes the set $\set{R}^{(1)}_3$ of all nonnegative rate tuples $(R_0,R_1,R_2)$ that satisfy 
 \begin{subequations}
 \begin{IEEEeqnarray}{rCl}
R_0&\leq &I(U_0;Y_1|X_1,X_2)+\Delta\\
R_0+R_1&\leq &I(U_0;Y_1|X_1,X_2)+\Delta+I_1 \\
R_0+R_2&\leq &I(U_0;Y_1|X_1,X_2)+\Delta+I_2 \\
R_0+R_2 &\leq &  I(U_0,U_2,X_1;Y_2|X_2)\\
&&- I(\hat{Y}_1;Y_1|U_0,U_2,X_1,X_2,Y_2)\\
R_0+R_1+R_2&\leq & I(U_0;Y_1|X_1,X_2)+I(U_0,U_2,X_1;Y_2|X_2)\nonumber\\
&&+\Delta- I(\hat{Y}_1;Y_1|U_0,U_2,X_1,X_2,Y_2)\nonumber\\
&&-I(U_1;U_2|U_0,X_1,X_2)\quad \\
R_0+R_1+R_2&\leq & I(U_0;Y_1|X_1,X_2)+\Delta\nonumber\\
&&+I_1+I_2-I(U_1;U_2|U_0,X_1,X_2)
\end{IEEEeqnarray}
for some pmf $P_{X_1}P_{X_2}{ P_{U_0U_1U_2|X_1X_2}}P_{\hat{Y}_1|U_0X_1X_2Y_1}P_{\hat{Y}_2|X_2Y_2}$ and function $X=f(U_0,U_1,U_2)$ such that 
\begin{IEEEeqnarray}{rCl}
&&I(\hat{Y}_1;Y_1|U_0,X_1,X_2,Y_{2})\leq R_{\text{fb},1}\\
&&I(\hat{Y}_2;Y_2|U_0,X_1,X_2,Y_{1})\leq R_{\text{fb},2}
\end{IEEEeqnarray}
 where
\begin{IEEEeqnarray}{rCl} 
&&\Delta=\min\{0,I(X_{2};Y_1|X_1)\!-\!I(\hat{Y}_{2};Y_{2}|U_0,X_1,X_2,Y_1)\}\nonumber\\
&&I_1=I(U_1;\hat{Y}_2,Y_1|U_0,X_1,X_2)\nonumber\\
&&\qquad+\min\{0,R_{\text{fb},2}\!-\!I(\hat{Y}_2;Y_2|U_0,X_1,X_2,Y_1)\} \nonumber\\
&&I_2=I(U_2;\hat{Y}_1,Y_2|U_0,X_1,X_2)\nonumber\\
&&\qquad+\min\{0,R_{\text{fb},1}\!-\!I(\hat{Y}_1;Y_1|U_0,X_1,X_2,Y_2)\}.\nonumber
\end{IEEEeqnarray}
  \end{subequations}

 \end{Theorem}

\begin{proof}
See Section \ref{sec:scheme2B}
\end{proof}

\begin{Remark}
 The region $\set{R}^{(2)}_3$ is also achievable by exchanging indices $1$ and $2$  in the above definition of $\set{R}^{(1)}_3$. The convex hull of the union of   $\set{R}^{(1)}_3$ and $\set{R}^{(2)}_3$ leads to a potentially larger rate region.
\end{Remark}





\section{Comparisons among $\set{R}_1$, $\set{R}_2$, $\set{R}_{\textnormal{Liang}}$ and $\set{R}_\textnormal{Wu}$}\label{Sec:compare}
We compare our  regions $\set{R}_1$, $\set{R}_2$ with the known rate regions   $\set{R}_{\textnormal{Liang}}$ and $\set{R}_\textnormal{Wu}$. Note that $\set{R}_{\textnormal{Liang}}$ is for RBCs without feedback and $\set{R}_\textnormal{Wu}$ is for BCs with feedback, while  $\set{R}_1$ and $\set{R}_2$ include both feedback communication and relay cooperation.  

 \subsection{ $\set{R}_{\textnormal{Liang}}$ versus $\set{R}_1$}\label{sec:Rliang}
 
 In Remark \ref{Remark:Liang}, it showed that our rate region $\set{R}_1$ includes Liang and Kramer's region $\set{R}_{\text{Liang}}$. In this subsection, we will prove that  when the feedback rate is sufficiently large,  this inclusion is strict for some channels, i.e. 
 \begin{IEEEeqnarray}{rCl}\label{eq:relationLiang}
 \set{R}_{\text{Liang}}\subset \set{R}_1.
\end{IEEEeqnarray}

Suppose $R_0=0$ for simplicity. To prove  \eqref{eq:relationLiang}, in view of Remark \ref{Remark:Liang}, it's sufficient to show  there exists some rate pair $(R^*_1,R^*_2)$ in $\set{R}_1$ lying strictly outside of $\set{R}_{\text{Liang}}$. 
Consider the corner point $(0,R^*_{2,\text{Liang}})$ on the boundary of  $\set{R}_{\text{Liang}}$ in  \eqref{rate:Liang}, where the transmitter spends all power to send message $M_2$ to Receiver 2, i.e., $U_1=
\emptyset$ and $U_2=X$. Thus, we have
\begin{subequations}\label{eq:regionDistDF}
\begin{IEEEeqnarray}{rCl}
R^*_{2,\text{Liang}}&\leq& I(X,X_1;Y_2)\\
R^*_{2,\text{Liang}}&\leq& I(U_0;Y_1|X_1)+I(X_1;Y_2|X_1,U_0) 
\end{IEEEeqnarray}
 for some pmf  $P_{XX_1U}$, which is the  partial decode-forward lower bound of relay channel \cite{Cover'79}. 
\end{subequations}

Now consider $\set{R}_1$ in \eqref{eq:regionSch1}. Let $R_0=R_1=0$ and $U_1=\emptyset$ and $U_2=X$,  then the marginal rate ${R}_2$ is achievable if 
 \begin{subequations}\label{eq:regionCFDF}
\begin{IEEEeqnarray}{rCl} 
R^*_{2,\text{Scheme1}}&\leq& I(X,X_1;Y_2) -I(\hat{Y}_1;Y_1|U_0,X,X_1,Y_2) \quad \\
R^*_{2,\text{Scheme1}}&\leq& I(U_0;Y_1|X_1)+I(X; \hat{Y}_1,Y_2|U_0,X_1)
\end{IEEEeqnarray}
for some pmf ${ P_{U_0X_1X}}P_{\hat{Y}_1|U_0X_1Y_1}$ satisfying 
\begin{IEEEeqnarray}{rCl} \label{eq:feedback}
&&I(\hat{Y}_1;Y_1|U_0,X_1)\leq R_{\text{fb},1}.
\end{IEEEeqnarray} 
 \end{subequations}
If feedback rate is sufficiently large such that  rate constraint \eqref{eq:feedback} is inactive, then \eqref{eq:regionCFDF} turns to be Gabbai and Bross's rate in \cite[Theorem 3]{Gabbai'06}. In their work, they evaluated the rates \eqref{eq:regionDistDF} and  \eqref{eq:regionCFDF} for  the  Gaussian  and $Z$ relay channels, and showed that 
$R^*_{2,\text{Scheme1}}>R^*_{2,\text{Liang}}$. In view of this fact and from Remark \ref{Remark:Liang}, we have
 \begin{Corollary}
$\set{R}_{\text{Liang}}\subset \set{R}_1$ holds when $R_{\text{fb},1}$ satisfies \eqref{eq:feedback}. 
\end{Corollary}

 \subsection{ $\set{R}_{\textnormal{Wu}}$ versus $\set{R}_2$}
  Remark \ref{Remark:WU} states that  $\set{R}_{\text{Wu}}\subseteq \set{R}_2$.
Here we  prove that $\set{R}_{\text{Wu}}\subset \set{R}_2$.  To prove the strict inclusion, we  follow  similar procedures  in Section \ref{sec:Rliang} and show that there exists some rate pair $(R^*_1,R^*_2)$ inside $\set{R}_2$   lying strictly outside of $ \set{R}_{\text{Wu}}$. 

Consider the corner point $(0,R^*_{2,\text{Wu}})$ on the boundary of $\set{R}_{\text{Wu}}$. From \eqref{eq:region_relay}, it's easy to check that  
\begin{IEEEeqnarray}{rCl}
R^*_{2,\text{Wu}}&\leq& I(X;Y_2) \label{eq:regionWu1}
\end{IEEEeqnarray}
 for some  $P_{X}$,
which is the  capacity of the link from the transmitter to Receiver 2.

Now consider $\set{R}_2$ in \eqref{eq:Scheme2A}. Let $R_0=R_1=0$ and $U_0=U_1=\hat{Y}_2=\emptyset$,  then the marginal rate ${R}_2$ is achievable if  
\begin{subequations}\label{eq:CF}
\begin{IEEEeqnarray}{rCl}
R^*_{2,\text{CF}}\leq &&I(X;\hat{Y}_1,Y_2|X_1)\\
R^*_{2,\text{CF}}\leq && I(X,X_1;Y_2)\!-\!I(\hat{Y}_1;Y_1|X,X_1,\!Y_2)
\end{IEEEeqnarray}
\end{subequations}
for some pmf $P_{X}P_{X_1}P_{\hat{Y}_1|X_1Y_1}$, which is  the compress-forward lower bound of the relay channel \cite{Cover'79}. It's well known that introducing a compress-forward relay to the point-to-point channel, such as Gaussian channel, can strictly increase the capacity \eqref{eq:regionWu1}. 
Thus, we have  
 \begin{Corollary}
$\set{R}_{\text{Wu}}\subset \set{R}_2$. 
\end{Corollary}

\subsection{Example}
Consider the Gaussian relay broadcast channel with perfect feedback from  Receiver 1 to the transmitter, see Fig. \ref{fig:relayFb}. The channel outputs are: 
\begin{IEEEeqnarray*}{rCl}
Y_1&=&g_{01}X+Z_1,\nonumber\\
Y_2&=&g_{02}X+g_{12}X_1+Z_2
\end{IEEEeqnarray*}
where $g_{01}$, $g_{02}$ and  $g_{12}$  are channel gains, and $Z_1\sim\set{N}(0,1)$ and $Z_2\sim\set{N}(0,1)$  are independent Gaussian noise variables. The  input power constraints  are $\mathbb{E}|X^2|\leq P$ and $\mathbb{E}|X^2_1|\leq P_1$.

\begin{figure}
\centering
\includegraphics[width=0.45\textwidth]{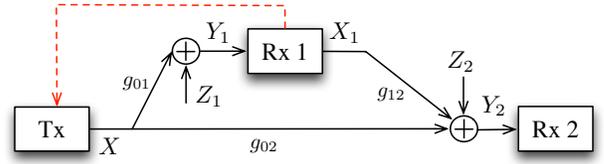}
\caption{Gaussian RBC with relay-transmitter feedback} \label{fig:relayFb}
\end{figure}

Table \ref{tab:compareRates} compares  $R^*_{2,\text{Liang}}$, $R^*_{2,\text{Scheme1}}$, $R^*_{2,\text{Wu}}$,  and $R^*_{2,\text{CF}}$, see \eqref{eq:regionDistDF}--\eqref{eq:CF}, for this channel with $g_{01}=1/d, g_{02}=1$, $g_{12}=1/|1-d|$, and  $P=5, P_1=1$. It can be seen that $R^*_{2,\text{Scheme1}}> R^*_{2,\text{CF}}>R^*_{2,\text{Liang}}>R^*_{2,\text{Wu}}$, which means that our rate regions $\set{R}_1$ and $\set{R}_2$ can strictly improve on $\set{R}_\text{Liang}$ and $\set{R}_{\text{Wu}}$, respectively.

\begin{table}
\begin{center}
\caption{Margnial rate $R^*_2$ achieved by   various coding schemes }
\begin{tabular}{lcccccc}
\toprule
$d$  & $R^*_{2,\text{Liang}}$ &  $R^*_{2,\text{Scheme1}}$  & $R^*_{2,\text{Wu}}$ & $R^*_{2,\text{CF}}$\\
\midrule
0.73& 1.6881 & 1.7069& 1.2925  & 1.6908  \\
0.74& 1.6703 & 1.7111 & 1.2925 & 1.6971  \\
0.75& 1.6529 & 1.7153 & 1.2925 & 1.7033 \\
0.76& 1.6358 & 1.7195 & 1.2925 & 1.7094\\
\bottomrule
 \label{tab:compareRates}
\end{tabular}
\end{center}
\end{table}

\section{Coding scheme for partially   cooperative BRCs with rate-limited feedback}\label{Sec:ParRBC}
In this section we present a block-Markov coding scheme for partially   cooperative BRCs with relay-transmitter and rate-limited feedback. Assume only Receiver 1 relays cooperative information $X_1$ without loss of generality. In the transmission,  a sequence of $B$ i.i.d message tuples $(m_{0,b},m_{1,b},m_{2,b})$, $b\in[1:B]$, are sent over $B+1$ blocks, each consisting of $n$ transmissions. 

Split  message $m_{k,b}$ into common and private parts: $m_{k,b}=(m_{c,k,b},m_{p,k,b})$, where $m_{c,k,b}\in[1:2^{nR_{c,k}}]$, $m_{p,k,b}\in[1:2^{nR_{p,k}}]$ and $R_k=R_{c,k}+R_{p,k}$. Define  $\textbf{m}_{c,b}:=(m_{0,b},m_{c,1,b}, m_{c,2,b})$ and $R_c:=R_0+R_{c,1}+R_{c,2}$.

In each block $b\in[1:B+1]$, after obtaining  feedback message $m_{\text{fb},1,b-1}$, the transmitter uses Marton's coding to  send $(\textbf{m}_{c,b}, \textbf{m}_{c,b-1}, m_{\text{fb},1,b-1})$ in the cloud centre $u_{0,b}^n$, and $m_{p,1,b}, m_{p,2,b}$ in two different satellites $u^n_{1,b},u^n_{2,b}$, respectively. Receiver 1  first jointly decodes $(\textbf{m}_{c,b}, {m}_{p,1,b})$, and then  compress its channel outputs $y^n_{1,b}$. Finally, it sends the compression message $m_{\text{fb},1,b}$ as feedback information  and   $x^n_{1,b+1}(\textbf{m}_{c,b},m_{\text{fb},1,b})$ as channel inputs in  next bock.   Receiver 2  uses backward decoding to jointly decode  $(\textbf{m}_{c,b-1}, {m}_{p,2,b},m_{\text{fb},1,b-1})$.  
 Note that  the transmitter knows  $(\textbf{m}_{c,b-1}, m_{\text{fb},1,b-1})$, from which it can reconstruct  Receiver 1's input $x^n_{1,b}$, thus  we superimpose $(u^n_{0,b},u^n_{1,b},u^n_{2,b})$ on  $x^n_{1,b}$ that attains cooperation between the transmitter and Receiver 1.
Coding is explained with the help of Table \ref{tab:Scheme1}.

\begin{table*}
\begin{center}
\caption{ Scheme 1  for partially cooperative BRCs with rate-limited feedback}
\begin{tabular}{>{\bfseries}ccccccc}
\toprule
Block &1& 2 & $\ldots$ & $b$ &$\cdots$\\
\midrule
$X_1$& $x_{1,1}^n(1,1)$ & $x_{1,2}^n(\textbf{m}_{c,1},m_{\text{fb},1,1})$ & $\ldots$   & $x_{1,b}^n(\textbf{m}_{c,b-1},m_{\text{fb},1,b-1})$  & $\ldots$ \\
$U_0$  & $u^n_{0,1}(\textbf{m}_{c,1}|1,1)$ & $u_{0,2}^n( \textbf{m}_{c,2}|\textbf{m}_{c,1},{m}_{\text{fb},1,1})$ & $\ldots$  &  $u_{0,b}^n( \textbf{m}_{c,b}|\textbf{m}_{c,b-1},{m}_{\text{fb},1,b-1})$   & $\cdots$\\

$U_k$ 	        & $u^n_{k,1}(m_{p,k,1},v_{k,1}|\textbf{m}_{c,1},1,1)$  & $u_{k,2}^n( m_{p,k,2},v_{k,2}|\textbf{m}_{c,2},\textbf{m}_{c,1},{m}_{\text{fb},1,1})$ & $\ldots$   &    $u_{k,b}^n( m_{p,k,b},v_{k,b}|\textbf{m}_{c,b},\textbf{m}_{c,b-1},{m}_{\text{fb},1,b-1})$ &     $\cdots$\\  

$\hat{Y}_1$      &  $\hat{y}_{1,1}^n(m_{\text{fb},k,1}|\textbf{m}_{c,1},1,1)$ & $\hat{y}_{1,2}^n(m_{\text{fb},1,2}|\textbf{m}_{c,2},\textbf{m}_{c,1}, m_{\text{fb},1,1})$ & $\ldots$   &    $\hat{y}_{1,b}^n(m_{\text{fb},1,b}|\textbf{m}_{c,b},\textbf{m}_{c,b-1}, m_{\text{fb},1,b-1})$ &     $\cdots$\\  
\midrule
${Y}_1$	& $(\hat{\textbf{m}}^{(1)}_{c,1},\hat{m}_{p,1,1},\hat{v}_{1,1})$    & $ (\hat{\textbf{m}}^{(1)}_{c,2},\hat{m}_{p,1,2},\hat{v}_{1,2})\rightarrow$      & $\ldots$ & $ (\hat{\textbf{m}}^{(1)}_{c,b},\hat{m}_{p,1,b},\hat{v}_{1,b})\rightarrow$    &         $\cdots$\\ 

${Y}_2$	& $(\hat{m}_{p,2,1},\hat{v}_{2,1})$     & $\leftarrow (\hat{\textbf{m}}^{(2)}_{c,1},\hat{m}_{p,2,2},\hat{v}_{2,2},\hat{m}_{\text{fb},1,1})$      & $\ldots$ & $\leftarrow (\hat{\textbf{m}}^{(2)}_{c,b\!-\!1},\hat{m}_{p,2,b},\hat{v}_{2,b},\hat{m}_{\text{fb},1,b-1})$    &         $\cdots$\\ 
\bottomrule
 \label{tab:Scheme1}
\end{tabular}
\end{center}
\end{table*}

\subsubsection{Code construction}
Fix pmf ${ P_{U_0U_1U_2X_1}}P_{\hat{Y}_1|U_0X_1Y_1}$ and a function $X=f(U_0,U_1,U_2)$.  For each block $b\in[1:B+1]$, randomly and independently generate $2^{n(R_c+\hat{R}_1)}$ sequences $x_{1,b}^n(\textbf{m}_{c,b\!-\!1},m_{\text{fb},1,b\!-\!1})\sim \prod^n_{i=1}P_{X_1}(x_{1,b,i})$,  $\textbf{m}_{c,b\!-\!1}\in[1:2^{nR_c}]$ and $m_{\text{fb},1,b-1}\in[1:2^{n\hat{R}_1}]$. For each $(\textbf{m}_{c,b\!-\!1},m_{\text{fb},1,b\!-\!1})$, randomly and independently generate $2^{n{R}_c}$ sequences $u_{0,b}^n( \textbf{m}_{c,b}|\textbf{m}_{c,b\!-\!1},{m}_{\text{fb},1,b-1})\sim \prod^n_{i=1}P_{U_0|X_1}(u_{0,b,i}|x_{1,b,i})$.  For each $(\textbf{m}_{c,b},\textbf{m}_{c,b-1},{m}_{\text{fb},1,b\!-\!1})$, randomly and independently generate $2^{n(
{R}_{p,k}+{R}'_k)}$ sequences $u_{k,b}^n( m_{p,k,b},v_{k,b}|\textbf{m}_{c,b},\textbf{m}_{c,b-1},{m}_{\text{fb},1,b-1})\sim \prod^n_{i=1}P_{U_k|U_0X_1}(u_{k,b,i}|u_{0,b,i},x_{1,b,i})$, $m_{p,k,b}\in[1:2^{n{R}_{p,k}}]$ and $v_{k,b}\in[1:2^{n{R}'_k}]$.   For each $(\textbf{m}_{c,b},\textbf{m}_{c,b-1},m_{\text{fb},1,b-1})$, randomly and independently generate $2^{n\hat{R}_1}$ sequences $\hat{y}_{1,b}^n(m_{\text{fb},1,b}|\textbf{m}_{c,b},\textbf{m}_{c,b\!-\!1}, m_{\text{fb},1,b\!-\!1}) \sim \prod^n_{i=1}P_{\hat{Y}_1|U_0X_1}(\hat{y}_{1,b,i}|u_{0,b,i},x_{1,b,i})$.

\subsubsection{Encoding}
  In each block $b\in[1:B+1]$, assume that the transmitter already knows ${{m}}_{\text{fb},1,b-1}$ through  feedback link.  It first looks for a pair of indices $(v_{1,b},v_{2,b})$ such that
\begin{IEEEeqnarray}{rCl}
&&\big( u^n_{1,b}(m_{p,1,b},v_{1,b}|\textbf{m}_{c,b},{\textbf{m}_{c,b-1}},{m}_{\text{fb},1,b\!-\!1}),\nonumber\\
&&~  u_{0,b}^n( \textbf{m}_{c,b}|\textbf{m}_{c,b\!-\!1},{m}_{\text{fb},1,b-1}),x_{1,b}^n(\textbf{m}_{c,b-1},m_{\text{fb},1,b-1}),\nonumber\\
&&~~  u^n_{2,b}(m_{p,2,b},v_{2,b}|\textbf{m}_{c,b},{\textbf{m}_{c,b-\!\!1},}{m}_{\text{fb},1,b\!-\!1}) \big) \!\in\!\mathcal{T}^n_{\epsilon}(P_{U_0U_1U_2X_1}).\nonumber
\end{IEEEeqnarray}
Then in block $b$ it sends $x^n_{b}$ with $x_{b,i}=f(u_{0,b,i},u_{1,b,i},u_{2,b,i})$.

By  covering  lemma \cite{book:gamal}, this is successful with high probability for sufficiently large $n$ if 
\begin{eqnarray}\label{eq:MartonEncS0}
{R}'_1+{R}'_2\geq I(U_1;U_2|U_0,X_1)
\end{eqnarray}

\subsubsection{Receiver 1's decoding}
In each block $b\in[1:B+1]$, Receiver 1 looks for  $(\hat{\textbf{m}}^{(1)}_{c,b}, \hat{m}_{p,1,b},\hat{v}_{1,b})$ such that
\begin{IEEEeqnarray*}{rCl}\label{eq:schem1Dec}
&&\big(x^n_{1,b}({\textbf{m}}_{c,b\!-\!1},m_{\text{fb},1,b-1}),y^n_{1,b},\nonumber\\
&&\quad u^n_{1,b}(\hat{m}_{p,1,b},\hat{v}_{1,b}|\hat{\textbf{m}}^{(1)}_{c,b},{\textbf{m}}_{c,b\!-\!1},{m}_{\text{fb},1,b\!-\!1}), \nonumber\\
&&\qquad u^n_{0,b}(\hat{\textbf{m}}^{(1)}_{c,b},{\textbf{m}}_{c,b\!-\!1},{m}_{\text{fb},1,b\!-\!1})\big)\in \mathcal{T}^n_{\epsilon} (P_{X_1U_0U_1Y_1}).
\end{IEEEeqnarray*}

It then compresses $y_{1,b}^{n}$ by finding $m_{\text{fb},1,b}$ satisfying 
\begin{IEEEeqnarray*}{rCl}\label{eq:Sch1CF}
&&\big(x^n_{1,b}({\textbf{m}}_{c,b\!-\!1},m_{\text{fb},1,b-1}),u_{0,b}^n( \textbf{m}_{c,b}|\textbf{m}_{c,b\!-\!1},{m}_{\text{fb},1,b-1}),y^n_{1,b},\nonumber\\
&&\qquad \hat{y}_{1,b}^n(m_{\text{fb},1,b}|\textbf{m}_{c,b},\textbf{m}_{c,b\!-\!1}, m_{\text{fb},1,b\!-\!1}) \big)\in \mathcal{T}^n_{\epsilon}(P_{\hat{Y}_1X_1U_0Y_1}).
\end{IEEEeqnarray*}

Finally, it sends $m_{\text{fb},1,b}$ as feedback message to the transmitter at  rate 
\begin{IEEEeqnarray}{rCl}
\hat{R}_1\leq R_{\text{fb},1},
\end{IEEEeqnarray}
 and  sends  $x^n_{1,b+1}({\textbf{m}}_{c,b},m_{\text{fb},1,b})$ as channel inputs  in   block $b+1$. 


By the independence of the codebooks, the Markov lemma \cite{book:gamal}, packing lemma \cite{book:gamal} and the induction on backward decoding,  these steps are successful with high probability if
\begin{subequations} 
\begin{IEEEeqnarray}{rCl}
R_{p,1}\!+\!R'_1 &<& I(U_1;Y_1|U_0,X_1)\\
R_{p,1}\!+\!R'_1\!+\!R_c&<&I(U_0,U_1;Y_1|X_1)\\
\hat{R}_1&>&I(\hat{Y}_1;Y_1|U_0,X_1)
\end{IEEEeqnarray}
\end{subequations}

\subsubsection{Receiver 2's decoding}
Receiver 2 performs backward decoding.  In each block $b\in[1:B+1]$, It looks for  $(\hat{\textbf{m}}^{(2)}_{c,b\!-\!1}, \hat{m}_{p,2,b},\hat{v}_{2,b},\hat{m}_{\text{fb},1,b\!-\!1})$ such that
\begin{IEEEeqnarray*}{rCl}\label{eq:schem1Dec}
&&\big(x^n_{1,b}(\hat{\textbf{m}}^{(2)}_{c,b\!-\!1},\hat{m}_{\text{fb},1,b-1}),\hat{y}_{1,b}^n(m_{\text{fb},1,b}|\textbf{m}_{c,b},\hat{\textbf{m}}^{(2)}_{c,b-1}, \hat{m}_{\text{fb},1,b\!-\!1}), \nonumber\\
&&\quad u^n_{2,b}(\hat{m}_{p,2,b},\hat{v}_{2,b}|{\textbf{m}}_{c,b},\hat{\textbf{m}}^{(2)}_{c,b-1},\hat{m}_{\text{fb},1,b\!-\!1}),y^n_{2,b}, \nonumber\\
&&\qquad u^n_{0,b}({\textbf{m}}_{c,b},\hat{\textbf{m}}^{(2)}_{c,b-1},\hat{m}_{\text{fb},1,b\!-\!1})\big)\in \mathcal{T}^n_{\epsilon} (P_{X_1U_0U_2Y_2\hat{Y}_1}).
\end{IEEEeqnarray*}

By the independence of the codebooks, the Markov lemma, packing lemma  and the induction on backward decoding,  these steps are successful with high probability if
\begin{subequations} \label{eq:rate0PDF}
\begin{IEEEeqnarray}{rCl}
R_{p,2}+R'_2 &<& I(U_2;Y_2,\hat{Y}_1|U_0,X_1)\\
R_{p,2}+R'_2+R_c+\hat{R}_1&<&I(U_0,U_2,X_1;Y_2)\nonumber\\
&&+I(\hat{Y}_1;U_2,Y_2|U_0,X_1)
\end{IEEEeqnarray}
\end{subequations}

Combine (\ref{eq:MartonEncS0}--\ref{eq:rate0PDF}) and  use Fourier-Motzkin elimination to eliminate $R'_1,R'_2, \hat{R}_1,\hat{R}_2$,  then we obtain Theorem \ref{Them:Scheme1}. 

\section{Achievable rates for fully cooperative RBC with rate-limited feedback}\label{sec:FullCooperation}

\begin{table*}[ht!]
\begin{center}
\caption{ Scheme 2A for fully cooperative BRCs with rate-limited feedback}
\begin{tabular}{>{\bfseries}lccccccc}
\toprule
Block &1&2 &$\ldots$ & $b$ &$\cdots$\\
\midrule
$X_k$& $x^n_{k,1}(1)$ & $x^n_{k,2}(m_{\text{fb},k,1})$ & $\ldots$   & $x^n_{k,b}(m_{\text{fb},k,b\!-\!1})$  & $\ldots$ \\
$U_0$  & $u^n_{0,1}(\textbf{m}_{c,1}|1,1)$ & $u^n_{0,2}(\textbf{m}_{c,2}|\textbf{m}_{\text{fb},1})$  & $\ldots$  &  $u^n_{0,b}(\textbf{m}_{c,b}|\textbf{m}_{\text{fb},b\!-\!1})$   & $\cdots$\\

$U_k$ 	        & $u^n_{k,1}(m_{p,k,1},v_{k,1}|\textbf{m}_{c,1},1,1)$ & $u^n_{k,2}(m_{p,k,2},v_{k,2}|\textbf{m}_{c,2},\textbf{m}_{\text{fb}, 1})$   & $\ldots$   &    $u^n_{k,b}(m_{p,k,b},v_{k,b}|\textbf{m}_{c,b},\textbf{m}_{\text{fb},b\!-\!1})$ &     $\cdots$\\  

$\hat{Y}_k$      &  $\hat{y}_{k,1}^n(m_{\text{fb},k,1}| 1)$ & $\hat{y}_{k,2}^n(m_{\text{fb},k,2}| m_{\text{fb},k,1})$   & $\ldots$   &    $\hat{y}_{k,b}^n(m_{\text{fb},k,b}| m_{\text{fb},k,b\!-\!1})$ &     $\cdots$\\  
\midrule
${Y}_1$	& $(\hat{\textbf{m}}^{(1)}_{c,1},\hat{m}_{p,1,1},\hat{v}_{1,1})$     &  $\leftarrow(\hat{\textbf{m}}^{(1)}_{c,2},\hat{m}_{p,1,2},\hat{v}_{1,2},\hat{m}_{\text{fb},2, 1})$      & $\ldots$ & $\leftarrow (\hat{\textbf{m}}^{(1)}_{c,b},\hat{m}_{p,1,b},\hat{v}_{1,b},\hat{m}_{\text{fb},2,b\!-\!1})$    &         $\cdots$\\ 

${Y}_2$	& $(\hat{\textbf{m}}^{(2)}_{c,1},\hat{m}_{p,2,1},\hat{v}_{2,1})$     & $\leftarrow (\hat{\textbf{m}}^{(2)}_{c,2},\hat{m}_{p,2,2},\hat{v}_{2,2},\hat{m}_{\text{fb},1, 1})$      & $\ldots$ & $\leftarrow (\hat{\textbf{m}}^{(2)}_{c,b},\hat{m}_{p,2,b},\hat{v}_{2,b},\hat{m}_{\text{fb},1,b\!-\!1})$    &         $\cdots$\\ 
\bottomrule
 \label{tab:CFNoisy}
\end{tabular}
\end{center}
\end{table*}
In this section we present two block-Markov coding schemes for fully   cooperative BRC with relay/receiver-transmitter and rate-limited feedback.
\subsection{Scheme 2A: Compress-forward relaying and backward decoding}\label{sec:scheme2A}
In this subsection we propose a block-Markov coding scheme where a sequence of $B$ i.i.d message tuples $(m_{0,b},m_{1,b},m_{2,b})$ are sent over $B+1$ blocks, each consisting of $n$ transmissions. Split the message $m_{k,b}$ in the same way as Section \ref{Sec:ParRBC} and define  $\textbf{m}_{\text{fb},b}:=({m}_{\text{fb},1,b},{m}_{\text{fb},2,b})$.


 In each block $b\in[1:B+1]$,  after obtaining  feedback messages $\textbf{m}_{\text{fb},b-1}$, the transmitter uses Marton's coding to  send $(\textbf{m}_{c,b}, \textbf{m}_{\text{fb},b-1})$ in the cloud centre $u_{0,b}^n$, and $m_{p,1,b}, m_{p,2,b}$ in two different satellites $u^n_{1,b},u^n_{2,b}$, respectively.   Receiver $k\in\{1,2\}$  first uses backward decoding to  decode  $(\textbf{m}_{c,b}, {m}_{p,k,b})$ and reconstructs the other receiver's compression message. Then, it compresses its  channel outputs $y^n_{k,b}$. Finally, it sends $m_{\text{fb},k,b}$ as feedback message and  $x^n_{k,b+1}(m_{\text{fb},k,b})$ as channel inputs  in  next bock. {Here  $(u^n_{0,b},u^n_{1,b},u^n_{2,b})$ are superimposed on  $(x^n_{1,b}, x^n_{2,b})$ that attains cooperation between the transmitter and the receivers. Coding is explained with the help of Table \ref{tab:CFNoisy}.}

\subsubsection{Code construction}
Fix pmf \[ P_{X_1}P_{X_2}{ P_{U_0U_1U_2|X_1X_2}}P_{\hat{Y}_1|X_1Y_1}P_{\hat{Y}_2|X_2Y_2}\] and a function $X=f(U_0,U_1,U_2)$.  For each block $b\in[1\!:\!B\!+\!1]$ and $k\in\{1,2\}$, randomly and independently generate $2^{n\hat{R}_k}$ sequences $x_{k,b}^n(m_{\text{fb},k,b-1})\sim \prod^n_{i=1}P_{X_k}(x_{k,b,i})$, $m_{\text{fb},k,b-1}\in[1:2^{n\hat{R}_k}]$.  For each $m_{\text{fb},k,b-1}$, randomly and independently generate $2^{n\hat{R}_k}$ sequences $\hat{y}_{k,b}^n(m_{\text{fb},k,b}| m_{\text{fb},k,b-1}) \sim \prod^n_{i=1}P_{\hat{Y}_k|X_k}(\hat{y}_{k,b,i}|x_{k,b,i})$.  For each $\textbf{m}_{\text{fb},b-1}$, randomly and independently generate $2^{n{R}_c}$ sequences $u_{0,b}^n( \textbf{m}_{c,b}|\textbf{m}_{\text{fb},b-1})\sim \prod^n_{i=1}P_{U_0|X_1X_2}(u_{0,b,i}|x_{1,b,i},x_{2,b,i})$, $\textbf{m}_{c,b}\in[1:2^{nR_c}]$. 
 For each $(\textbf{m}_{c,b},\textbf{m}_{\text{fb},b-1})$, randomly and independently generate $2^{n({R}_{p,k}+{R}'_k)}$ sequences $u_{k,b}^n( m_{p,k,b},v_{k,b}|\textbf{m}_{c,b},\textbf{m}_{\text{fb},b-1})\sim \prod^n_{i=1}P_{U_k|U_0X_1X_2}(u_{k,b,i}|u_{0,b,i},x_{1,b,i},x_{2,b,i})$, $m_{p,k,b}\in[1:2^{n{R}_{p,k}}]$ and $v_{k,b}\in[1:2^{n{R}'_k}]$. 


\subsubsection{Encoding}
  In each block $b\in[1:B+1]$, assume that the transmitter already knows ${\textbf{m}}_{\text{fb},b-1}$ through  feedback links.  It first looks for a pair of indices $(v_{1,b},v_{2,b})$ such that
\begin{IEEEeqnarray}{rCl}
&&\big(u^n_{0,b}(\textbf{m}_{c,b}|\textbf{m}_{\text{fb},b\!-\!1}),x^n_{1,b}(m_{\text{fb},1,b\!-\!1}),\nonumber\\
&&\quad u^n_{1,b}(m_{p,1,b},v_{1,b}|\textbf{m}_{c,b},\textbf{m}_{\text{fb},b\!-\!1}),x^n_{2,b}(m_{\text{fb},2,b\!-\!1}),\nonumber\\
&&\qquad u^n_{2,b}(m_{p,2,b},v_{2,b}|\textbf{m}_{c,b},\textbf{m}_{\text{fb},b\!-\!1}) \big) \!\in\!\mathcal{T}^n_{\epsilon}(P_{U_0U_1U_2X_1X_2})\nonumber
\end{IEEEeqnarray}
Then in block $b$ it sends $x^n_{b}$ with $x_{b,i}=f(u_{0,b,i},u_{1,b,i},u_{2,b,i})$.

By  covering  lemma, this is successful with high probability for sufficiently large $n$ if 
\begin{eqnarray}\label{eq:MartonEnc}
{R}'_1+{R}'_2\geq I(U_1;U_2|U_0,X_1,X_2).
\end{eqnarray}

\subsubsection{Decoding}
Both receivers perform backward decoding and compress-forward strategy.  In each block $b\in[1:B+1]$, Receiver 1 looks for  $(\hat{\textbf{m}}^{(1)}_{c,b}, \hat{m}_{p,1,b},\hat{v}_{1,b},\hat{m}_{\text{fb},2,b\!-\!1})$ such that
\begin{IEEEeqnarray*}{rCl}\label{eq:schem1Dec}
&&\big(x^n_{1,b}(m_{\text{fb},1,b-1}),x^n_{2,b}(\hat{m}_{\text{fb},2,b-1}),\hat{y}_{2,b}^n(m_{\text{fb},2,b}| \hat{m}_{\text{fb},2,b\!-\!1}),\nonumber\\
&&\quad u^n_{1,b}(\hat{m}_{p,1,b},\hat{v}_{1,b}|\hat{\textbf{m}}_{c,b},{m}_{\text{fb},1,b\!-\!1},\hat{m}_{\text{fb},2,b\!-\!1}),y^n_{1,b}, \nonumber\\
&&\qquad u^n_{0,b}(\hat{\textbf{m}}_{c,b}|{m}_{\text{fb},1,b\!-\!1},\hat{m}_{\text{fb},2,b\!-\!1})\big)\in \mathcal{T}^n_{\epsilon} (P_{X_1X_2U_0U_1Y_1\hat{Y}_2}).
\end{IEEEeqnarray*}

It then compresses $y_{1,b}^{n}$ by finding a unique index $m_{\text{fb},1,b}$ such that 
\begin{IEEEeqnarray*}{rCl}\label{eq:Sch1CF}
&&\big(x^n_{1,b}(m_{\text{fb},1,b\!-\!1}),\hat{y}_{1,b}^n(m_{\text{fb},1,b}| m_{\text{fb},1,b\!-\!1}),y^n_{1,b} \big)\in \mathcal{T}^n_{\epsilon}(P_{\hat{Y}_1X_1Y_1}).
\end{IEEEeqnarray*}

Finally, in  block $b+1$ it sends  $x_{1,b+1}^n(m_{\text{fb},1,b})$ as channel input and forwards $m_{\textbf{fb},1,b}$ through the feedback link at rate:
\begin{IEEEeqnarray}{rCl}\label{eq:feedbackrate}
\hat{R}_1\leq R_{\text{fb},1}. 
\end{IEEEeqnarray}
Receiver 2 performs in a similar way with exchanging indices of 1 and 2 in above steps.

By the independence of the codebooks, the Markov lemma, packing lemma and the induction on backward decoding,  these steps are successful with high probability if
\begin{subequations} \label{eq:rate2APDF}
\begin{IEEEeqnarray}{rCl}
\hat{R}_1&>&I(\hat{Y}_1;Y_1|X_1)\\
\hat{R}_2&>&I(\hat{Y}_2;Y_2|X_2)\\
R_{p,1}\!+\!R'_1 &<& I(U_1;Y_1,\hat{Y}_2|U_0,X_1,X_2)\\
R_{p,2}\!+\!R'_2 &<& I(U_2;Y_2,\hat{Y}_1|U_0,X_1,X_2)\\
R_{p,1}\!+\!R'_1\!+\!R_c&<&I(U_0,U_1;\hat{Y}_2,Y_1|X_1,X_2)\\
R_{p,2}\!+\!R'_2\!+\!R_c&<&I(U_0,U_2;\hat{Y}_1,Y_2|X_1,X_2)\\
R_{p,1}\!+\!R'_1\!+\!R_c\!+\!\hat{R}_2&<&I(U_0,U_1,X_2;Y_1|X_1)\nonumber\\
&&+I(\hat{Y}_2;U_0,U_2,Y_1,X_1|X_2)\qquad\\
R_{p,2}\!+\!R'_2\!+\!R_c\!+\!\hat{R}_1&<&I(U_0,U_2,X_1;Y_2|X_2)\nonumber\\
&&+I(\hat{Y}_1;U_0,U_1,Y_2,X_2|X_1).\quad
\end{IEEEeqnarray}
\end{subequations}


Combine (\ref{eq:MartonEnc}--\ref{eq:rate2APDF}) and  use Fourier-Motzkin elimination to eliminate $R'_1,R'_2, \hat{R}_1,\hat{R}_2$, then  we obtain Theorem \ref{Them:Scheme2A}.

\subsection{Scheme 2B: Hybrid relaying strategy and sliding-window decoding}\label{sec:scheme2B}

\begin{table*}[ht!]
\begin{center}
\caption{ Scheme 2B for  fully cooperative BRCs with rate-limited feedback}
\begin{tabular}{>{\bfseries}lccccccc}
\toprule
Block &1& 2 & $\ldots$ & $b$ &$\cdots$\\
\midrule
$X_1$& $x^n_{1,1}(\textbf{1}_{[3]}, 1)$  & $x^n_{1,2}({\textbf{m}}_{c,1},m_{\text{fb},1,1})$ & $\ldots$   & $x^n_{1,b}({\textbf{m}}_{c,b-1},m_{\text{fb},1,b\!-\!1})$  & $\ldots$ \\
$X_2$& $x^n_{2,1}(1)$ & $x^n_{2,2}(m_{\text{fb},2,1})$  & $\ldots$   & $x^n_{2,b}(m_{\text{fb},2,b\!-\!1})$  & $\ldots$ \\
$U_0$  & $u^n_{0,1}(\textbf{m}_{c,1}|\textbf{1}_{[3]},\textbf{1}_{[2]})$ & $u^n_{0,2}({\textbf{m}}_{c,2}|\textbf{m}_{c,1},\textbf{m}_{\text{fb},1})$ & $\ldots$  &  $u^n_{0,b}({\textbf{m}}_{c,b}|\textbf{m}_{c,b\!-\!1},\textbf{m}_{\text{fb},b\!-\!1})$   & $\cdots$\\

$U_k$ 	        & $u^n_{k,1}(m_{k,1},v_{k,1}|\textbf{m}_{c,1},\textbf{1}_{[3]},\textbf{1}_{[2]})$  & $u^n_{k,2}(m_{k,2},v_{k,2}|\textbf{m}_{c,2},\textbf{m}_{c,1},\textbf{m}_{\text{fb},1})$ & $\ldots$   &    $u^n_{k,b}(m_{k,b},v_{k,b}|\textbf{m}_{c,b},\textbf{m}_{c,b\!-\!1},\textbf{m}_{\text{fb},b\!-\!1})$ &     $\cdots$\\  

$\hat{Y}_1$      &  $\hat{y}_{1,1}^n(m_{\text{fb},1,1},j_{1,1}| \textbf{m}_{c,1},\textbf{1}_{[3]},\textbf{1}_{[2]})$ & $\hat{y}_{1,2}^n(m_{\text{fb},1,2},{ j_{1,2}}| \textbf{m}_{c,2},{ \textbf{m}_{c,1}},\textbf{m}_{\text{fb},1})$ & $\ldots$   &    $\hat{y}_{1,b}^n(m_{\text{fb},1,b},{ j_{1,b}}| \textbf{m}_{c,b},{ \textbf{m}_{c,b\!-\!1}},\textbf{m}_{\text{fb},b\!-\!1})$ &     $\cdots$\\  

$\hat{Y}_2$   & $\hat{y}_{2,2}^n(m_{\text{fb},2,2},j_{2,2}| m_{\text{fb},2,1})$   &  $\hat{y}_{2,1}^n(m_{\text{fb},2,1},j_{2,1}| 1)$  & $\ldots$   &    $\hat{y}_{2,b}^n(m_{\text{fb},2,b},j_{2,b}| m_{\text{fb},2,b\!-\!1})$ &     $\cdots$\\  
\midrule
${Y}_1$	& $\hat{\textbf{m}}^{(1)}_{c,1}\rightarrow$    & $(\hat{\textbf{m}}^{(1)}_{c,2},\hat{m}_{\text{fb},2,1}),(j_{2,1}, \hat{m}_{p,1,1},\hat{v}_{1,1})\to$       & $\ldots$ & $(\hat{\textbf{m}}^{(1)}_{c,b},\hat{m}_{\text{fb},2,b\!-\!1}),(j_{2,b\!-\!1}, \hat{m}_{p,1,b\!-\!1},\hat{v}_{1,b\!-\!1})\to$    &         $\cdots$\\ 

${Y}_2$	& $(\hat{m}_{p,2,1},\hat{v}_{2,1},\hat{j}_{1,b})$          &  $\leftarrow (\hat{\textbf{m}}^{(2)}_{c,1},\hat{m}_{p,2,2},\hat{v}_{2,2},\hat{m}_{\text{fb},1,1},\hat{j}_{1,2})$       & $\ldots$ & $\leftarrow (\hat{\textbf{m}}^{(2)}_{c,b-1},\hat{m}_{p,2,b},\hat{v}_{2,b},\hat{m}_{\text{fb},1,b\!-\!1},\hat{j}_{1,b})$    &         $\cdots$\\ 
\bottomrule
 \label{tab:PDFNoisy}
\end{tabular}
\end{center}
\end{table*}

In Scheme 2A both receivers apply  compress-forward. 
In this subsection, we propose a coding scheme where  one of the two receivers, called  Receiver 1 without loss of generality,  applies a hybrid relaying strategy that combines partially decode-forward and compress-forward.   More specifically,  Receiver 1 first decodes the cloud center containing $(\textbf{m}_{c,b}, {m}_{\text{fb},2,b\!-\!1})$, then reconstructs Receiver 2's compression outputs $\hat{y}^n_{2,b\!-\!1}$ and decodes ${m}_{p,1,b\!-\!1}$ based on the enhanced outputs $(\hat{y}^n_{2,b\!-\!1},{y}^n_{1,b\!-\!1})$.  Finally it compresses  $y_{1,b}^n$, and sends the compression message ${m}_{\text{fb},1,b}$ as feedback  and $x^n_{1,b+1}(\textbf{m}_{c,b},{m}_{\text{fb},1,b})$ as channel inputs in  block $b+1$. Note that Receiver 1  needs to decode $\textbf{m}_{c,b}$ before sending $x^n_{1,b+1}$, thus it  has to use sliding-window decoding instead of backward decoding. The transmitter and the other receiver perform similarly as Scheme 1A. Coding is explained with the help of Table \ref{tab:PDFNoisy}.

\subsubsection{Code construction}
Fix pmf $P_{X_1}P_{X_2}{ P_{U_0U_1U_2|X_1X_2}}P_{\hat{Y}_1|U_0X_1X_2Y_1}P_{\hat{Y}_2|X_2Y_2}$ and a function $X=f(U_0,U_1,U_2)$.  For each block $b\in[1\!:\!B\!+\!1]$, randomly and independently generate $2^{n(R_c+\hat{R}_1)}$ sequences $x_{1,b}^n(\textbf{m}_{c,b-1},m_{\text{fb},1,b-1})\sim \prod^n_{i=1}P_{X_1}(x_{1,b,i})$, for $\textbf{m}_{c,b-1}\in[1:2^{nR_c}]$ and $m_{\text{fb},1,b-1}\in[1:2^{n\hat{R}_1}]$.  Randomly and independently generate $2^{n\hat{R}_2}$ sequences $x_{2,b}^n(m_{\text{fb},2,b-1})\sim \prod^n_{i=1}P_{X_2}(x_{2,b,i})$, for $m_{\text{fb},2,b-1}\in[1:2^{n\hat{R}_k}]$.  For each $(\textbf{m}_{c,b\!-\!1},\textbf{m}_{\text{fb},b-1})$, randomly and independently generate $2^{n{R}_c}$ sequences $u_{0,b}^n( \textbf{m}_{c,b}|\textbf{m}_{c,b\!-\!1},\textbf{m}_{\text{fb},b-1})\sim \prod^n_{i=1}P_{U_0|X_1X_2}(u_{0,b,i}|x_{1,b,i},x_{2,b,i})$, $\textbf{m}_{c,b}\in[1:2^{nR_c}]$. For each $(\textbf{m}_{c,b},\textbf{m}_{c,b-1},\textbf{m}_{\text{fb},b\!-\!1})$, randomly and independently generate $2^{n(
{R}_{p,k}+{R}'_k)}$ sequences $u_{k,b}^n( m_{p,k,b},v_{k,b}|\textbf{m}_{c,b},\textbf{m}_{c,b-1},\textbf{m}_{\text{fb},b-1})\sim \prod^n_{i=1}P_{U_k|U_0X_1X_2}(u_{k,b,i}|u_{0,b,i},x_{1,b,i},x_{2,b,i})$, $m_{p,k,b}\in[1:2^{n{R}_{p,k}}]$ and $v_{k,b}\in[1:2^{n{R}'_k}]$. For each $m_{\text{fb},2,b-1}$, randomly and independently generate $2^{n(\hat{R}_2+\tilde{R}_2)}$ sequences $\hat{y}_{2,b}^n(m_{\text{fb},2,b},j_{2,b}| m_{\text{fb},2,b-1}) \sim \prod^n_{i=1}P_{\hat{Y}_2|X_2}(\hat{y}_{2,b,i}|x_{2,b,i})$, ${j}_{2,b}\in[1:2^{n\tilde{R}_2}]$.  For each $(\textbf{m}_{c,b},\textbf{m}_{c,b-1},\textbf{m}_{\text{fb},b-1})$, randomly and independently generate $2^{n(\hat{R}_1+\tilde{R}_1)}$ sequences $\hat{y}_{1,b}^n(m_{\text{fb},1,b},j_{1,b}|\textbf{m}_{c,b},\textbf{m}_{c,b\!-\!1}, \textbf{m}_{\text{fb},b\!-\!1}) \sim \prod^n_{i=1}P_{\hat{Y}_1|U_0X_1X_2}(\hat{y}_{1,b,i}|u_{0,b,i},x_{1,b,i},x_{2,b,i})$, ${j}_{1,b}\in[1:2^{n\tilde{R}_1}]$.


\subsubsection{Encoding}
  In each block $b\in[1:B+1]$, assume that the transmitter already knows ${\textbf{m}}_{\text{fb},b-1}$ through  feedback links.  It first looks for a pair of indices $(v_{1,b},v_{2,b})$ such that
\begin{IEEEeqnarray}{rCl}
&&\big(u^n_{0,b}({\textbf{m}}_{c,b}|\textbf{m}_{c,b\!-\!1},\textbf{m}_{\text{fb},b\!-\!1}),x_{1,b}^n({\textbf{m}}_{c,b\!-\!1},m_{\text{fb},1,b-1}),\nonumber\\
&&~  u^n_{1,b}(m_{p,1,b},v_{1,b}|\textbf{m}_{c,b},{\textbf{m}}_{c,b\!-\!1},\textbf{m}_{\text{fb},b\!-\!1}),x_{2,b}^n(m_{\text{fb},2,b-1}),\nonumber\\
&&~~ u^n_{2,b}(m_{p,2,b},v_{2,b}|\textbf{m}_{c,b},{\textbf{m}}_{c,b\!-\!1},\textbf{m}_{\text{fb},b\!-\!1}) \big) \!\in\!\mathcal{T}^n_{\epsilon}(P_{U_0U_1U_2X_1X_2}).\nonumber
\end{IEEEeqnarray}
Then in block $b$ it sends $x^n_{b}$ with $x_{b,i}=f(u_{0,b,i},u_{1,b,i},u_{2,b,i})$.

By  covering  lemma, this is successful with high probability for sufficiently large $n$ if 
\begin{eqnarray}\label{eq:MartonEncScheme2}
{R}'_1+{R}'_2\geq I(U_1;U_2|U_0,X_1,X_2).
\end{eqnarray}

\subsubsection{Receiver 1's Decoding}  In each block $b\in[1:B+1]$, Receiver 1 first decodes cloud centre $u^n_{0,b}$ by looking for $(\hat{\textbf{m}}^{(1)}_{c,b}, \hat{m}_{\text{fb},2,b\!-\!1})$ such that
\begin{IEEEeqnarray*}{rCl}
&&\big(u^n_{0,b}(\hat{\textbf{m}}^{(1)}_{c,b}|{\textbf{m}}_{c,b\!-\!1},{m}_{\text{fb},1,b\!-\!1},\hat{m}_{\text{fb},2,b\!-\!1}),x^n_{2,b}(\hat{m}_{\text{fb},2,b\!-\!1}), \nonumber\\
&&\qquad x^n_{1,b}({\textbf{m}}_{c,b-1},m_{\text{fb},1,b\!-\!1}),y^n_{1,b}\big)\in \mathcal{T}^n_{\epsilon}(P_{U_0X_1X_2Y_1}).
\end{IEEEeqnarray*}

It then decodes $(\hat{y}^n_{2,b-1},u^n_{1,b-1})$ by looking for $(\hat{j}_{2,b\!-\!1}, \hat{m}_{1,b\!-\!1},\hat{v}_{1,b\!-\!1})$ such that
\begin{IEEEeqnarray*}{rCl}
&&\big(u^n_{0,b\!-\!1}({\textbf{m}}_{c,b\!-\!1}|\textbf{m}_{c,b\!-\!2},\textbf{m}_{\text{fb},b\!-\!2}), x_{1,b-1}^n({\textbf{m}}_{c,b\!-\!2},m_{\text{fb},1,b-2}), \nonumber\\
&&\quad \hat{y}_{2,b-1}^n(m_{\text{fb},2,b-1},\hat{j}_{2,b-1}|m_{\text{fb},1,b-1}),\nonumber\\
&&\quad\quad u^n_{1,b-1}(\hat{m}_{p,1,b-1},\hat{v}_{1,b-1}|\textbf{m}_{c,b},\textbf{m}_{c,b-1}, \textbf{m}_{\text{fb},b-1}),\\ 
&&\qquad\qquad x^n_{2,b\!-\!1}(m_{\text{fb},2,b-2}),y^n_{2,b-1}\big)\in \mathcal{T}^n_{\epsilon}(P_{U_0U_1X_1X_2\hat{Y}_2Y_1}).
\end{IEEEeqnarray*}

Then, it compresses $y_{1,b}^{n}$ by looking for a unique pair $(m_{\text{fb},1,b},j_{1,b})$ such that
\begin{IEEEeqnarray*}{rCl}\label{eq:Sch1CF}
&&\big(u^n_{0,b}({\textbf{m}}_{c,b}|\textbf{m}_{c,b\!-\!1},\textbf{m}_{\text{fb},b\!-\!1}), x_{1,b}^n({\textbf{m}}_{c,b\!-\!1},m_{\text{fb},1,b-1}),\nonumber\\
&&\qquad \hat{y}_{1,b}^n(m_{\text{fb},1,b},j_{1,b}|\textbf{m}_{c,b}, \textbf{m}_{c,b-1}, \textbf{m}_{\text{fb},b-1}),\nonumber\\
&&\qquad\qquad x^n_{2,b}(m_{\text{fb},2,b-1}),y^n_{1,b} \big)\in \mathcal{T}^n_{\epsilon}(P_{\hat{Y}_1U_0X_1X_2Y_1}).
\end{IEEEeqnarray*}

Finally, in  block $b+1$ it sends  $x_{1,b+1}^n(\textbf{m}_{c,b},m_{\text{fb},1,b})$ as channel input and forwards $m_{\textbf{fb},1,b}$ through the feedback link at rate:
\begin{IEEEeqnarray}{rCl}\label{eq:feedbackrate}
\hat{R}_1\leq  R_{\text{fb},1}. 
\end{IEEEeqnarray}

By the independence of the codebooks, the Markov lemma, packing lemma and the induction on backward decoding,  these steps are successful with high probability if
\begin{subequations} \label{eq:rate1PDF}
\begin{IEEEeqnarray}{rCl}
R_c&<&I(U_0;Y_1|X_1,X_2)\\
R_c+\hat{R}_2&<&I(U_0,X_2;Y_1|X_1)\\
\tilde{R}_2&<&I(\hat{Y}_2;U_0,X_1,Y_1|X_2)\\
R_{p,1}+R'_1 &<&I(U_1;Y_1,\hat{Y}_2|U_0,X_1,X_2)\\
R_{p,1}+R'_1+\tilde{R}_2 &<&I(U_1;Y_1,\hat{Y}_2|U_0,X_1,X_2)\nonumber\\
&&+I(\hat{Y}_2;U_0,X_1,Y_1|X_2)\\
\hat{R}_1+\tilde{R}_1&>&I(\hat{Y}_1;Y_1|U_0,X_1,X_2).
\end{IEEEeqnarray}
\end{subequations}

\subsubsection{Receiver 2's Decoding}
Receiver 2 performs backward decoding.  In each block $b\in[1:B+1]$, Receiver 2 looks for  $(\hat{\textbf{m}}^{(2)}_{c,b-1}, \hat{m}_{p,2,b},\hat{v}_{2,b},\hat{m}_{\text{fb},1,b\!-\!1},\hat{j}_{1,b})$ such that
\begin{IEEEeqnarray*}{rCl}\label{eq:schem1Dec}
&&\big(u^n_{0,b}({\textbf{m}}_{c,b}|\hat{\textbf{m}}^{(2)}_{c,b\!-\!1},\hat{m}_{\text{fb},1,b\!-\!1},{m}_{\text{fb},2,b\!-\!1}),x^n_{2,b}(m_{\text{fb},2,b-1}),\nonumber\\
&& \quad \hat{y}_{1,b}^n(m_{\text{fb},1,b},\hat{j}_{1,b}|{\textbf{m}}_{c,b},\hat{\textbf{m}}^{(2)}_{c,b-1}, \hat{m}_{\text{fb},1,b-1}),y^n_{2,b},\nonumber\\
&&\qquad u^n_{2,b}(\hat{m}_{p,2,b},\hat{v}_{2,b}|{\textbf{m}}_{c,b},\hat{\textbf{m}}^{(2)}_{c,b\!-\!1},\hat{m}_{\text{fb},1,b\!-\!1},{m}_{\text{fb},2,b\!-\!1}), \nonumber\\
&&\qquad\quad x^n_{1,b}(\hat{\textbf{m}}^{(2)}_{c,b\!-\!1},\hat{m}_{\text{fb},1,b\!-\!1})\big)\in \mathcal{T}^n_{\epsilon} (P_{X_1X_2U_0U_1Y_2\hat{Y}_1}).
\end{IEEEeqnarray*}

Also, it  compresses $y_{2,b}^{n}$ by looking for  a unique pair $(m_{\text{fb},2,b},j_{2,b})$ such that .
\begin{IEEEeqnarray*}{rCl}\label{eq:Sch1CF}
&&\big(x^n_{2,b}(m_{\text{fb},1,b\!-\!1}),y^n_{2,b},\\
&&\qquad \hat{y}_{2,b}^n(m_{\text{fb},2,b},j_{2,b}| m_{\text{fb},2,b\!-\!1}) \big)\in \mathcal{T}^n_{\epsilon}(P_{\hat{Y}_2X_2Y_2}).
\end{IEEEeqnarray*}
Finally, in  block $b+1$ it sends  $x_{2,b+1}^n(m_{\text{fb},2,b})$ as channel input and forwards $m_{\textbf{fb},2,b}$ through the feedback link at rate:
\begin{IEEEeqnarray}{rCl}\label{eq:feedbackrate}
\hat{R}_2\leq R_{\text{fb},2}. 
\end{IEEEeqnarray}

By the independence of the codebooks, the Markov lemma, packing lemma and the induction on backward decoding,  these steps are successful with high probability if
\begin{subequations} \label{eq:rate1PDFScheme2}
\begin{IEEEeqnarray}{rCl}
\tilde{R}_1 &< & I(\hat{Y}_1;Y_2,U_2|U_0,X_1,X_2)\\
R_{p,2}\!+\!{R}'_2&<&I(U_2;Y_2,\hat{Y}_1|U_0,X_1,X_2)\\
R_{p,2}\!+\!{R}'_2\!+\!\tilde{R}_1&<&I(U_2;Y_2,\hat{Y}_1|U_0,X_1,X_2)\nonumber\\
&&+ I(\hat{Y}_1;Y_2|U_0,X_1,X_2)\\
R_c\!+\!\hat{R}_1\!+\!R_{p,2}\!+\!{R}'_2+\tilde{R}_1&<&I(\hat{Y}_1;Y_2,U_2|U_0,X_1,X_2)\qquad \nonumber\\
&& +I(U_0,U_2,X_1;Y_2|X_2)\\
\hat{R}_2+\tilde{R}_2&>&I(\hat{Y}_2;Y_2|X_2).
\end{IEEEeqnarray}
\end{subequations}

Combine (\ref{eq:MartonEncScheme2}--\ref{eq:rate1PDFScheme2}) and  use Fourier-Motzkin elimination to eliminate $R'_1,R'_2, \hat{R}_1,\hat{R}_2,\tilde{R}_1,\tilde{R}_2$, then  we obtain Theorem \ref{Them:Scheme2B}.

\section{Conclusion}\label{Sec:conclusion}
In this paper, we studied partially   and fully cooperative RBCs with relay/receiver-transmitter and rate-limited  feedback. New coding schemes have been proposed to improve on the known rate regions that consider either feedback or relay cooperation, but not both. Specifically, our first rate region strictly improves on Liang and Kramer's region for the partially cooperative RBCs without feedback, and our second rate region strictly improves Wu and Wigger's region for the BCs with feedback but in the absence of relay cooperation. These two results together demonstrates that using feedback and relay simultaneously is a powerful tool to  improve the rate performance of networks.

\end{document}